\documentclass[conference,letterpaper]{IEEEtran}
\makeatletter
\newcommand{\hathat}[1]{%
\begingroup%
  \let\macc@kerna\z@%
  \let\macc@kernb\z@%
  \let\macc@nucleus\@empty%
  \hat{\mathchoice%
    {\raisebox{.2ex}{\vphantom{\ensuremath{\displaystyle #1}}}}%
    {\raisebox{.2ex}{\vphantom{\ensuremath{\textstyle #1}}}}%
    {\raisebox{.16ex}{\vphantom{\ensuremath{\scriptstyle #1}}}}%
    {\raisebox{.14ex}{\vphantom{\ensuremath{\scriptscriptstyle #1}}}}%
    \smash{\hat{#1}}}%
\endgroup%
}
\makeatother

\usepackage[utf8]{inputenc} 
\usepackage[T1]{fontenc}
\usepackage{url}
\usepackage{ifthen}
\usepackage{cite}
\usepackage[cmex10]{amsmath} 
\usepackage{authblk}
\usepackage{hyperref}
\hypersetup{
 colorlinks,
 linkcolor={blue!100!black},
 citecolor={blue!100!black},
 urlcolor={blue!80!black} 
}
\usepackage{amsthm,amssymb,amsfonts} 
\usepackage{algorithmic}
\usepackage{graphicx}
\usepackage{textcomp}
\usepackage{xcolor}
\usepackage{tikz}
\usepackage{ifthen}

\newtheorem{axiom}{Axiom}
\newtheorem{remark}{Remark}
\newtheorem{definition}{Definition}
\newtheorem{lemma}{Lemma}
\newtheorem{corollary}{Corollary}
\newtheorem{property}{Property}

\newboolean{IncludeAppendix}
\setboolean{IncludeAppendix}{true} 

\begin{document}
\title{Explicit Formula for\\Partial Information Decomposition} 
\author{\textbf{Aobo Lyu}$^\star$, \textbf{Andrew Clark}$^{\star}$, and \textbf{Netanel Raviv}$^\dagger$\\
$^\star$Department of Electrical and Systems Engineering, Washington University in St. Louis, St. Louis, MO, USA\\
		$^\dagger$Department of Computer Science and Engineering, Washington University in St. Louis, St. Louis, MO, USA\\
   \texttt{aobo.lyu@wustl.edu}, \texttt{andrewclark@wustl.edu}, \texttt{netanel.raviv@wustl.edu}}

\maketitle

\begin{abstract}
Mutual information between two random variables is a well-studied notion, whose understanding is fairly complete. Mutual information between one random variable and a pair of other random variables, however, is a far more involved notion. Specifically, Shannon's mutual information does not capture fine-grained interactions between those three variables, resulting in limited insights in complex systems. To capture these fine-grained interactions, in 2010 Williams and Beer proposed to decompose this mutual information to \textit{information atoms}, called unique, redundant, and synergistic, and proposed several operational axioms that these atoms must satisfy. In spite of numerous efforts, a general formula which satisfies these axioms has yet to be found. Inspired by Judea Pearl's do-calculus, we resolve this open problem by introducing the \textit{do-operation}, an operation over the variable system which sets a certain marginal to a desired value, which is distinct from any existing approaches. Using this operation, we provide the first explicit formula for calculating the information atoms so that Williams and Beer's axioms are satisfied, as well as additional properties from subsequent studies in the field.
\end{abstract}
\pagestyle{empty}
\section{Introduction}

Since its inception by Claude Shannon~\cite{shannon2001mathematical}, mutual information has remained a pivotal measure in information theory, which finds extensive applications across multiple other domains. 
Extending mutual information to multivariate systems has attracted significant academic interest, but no widely agreed upon generalization exists to date. For instance, the so-called \textit{interaction information}~\cite{watanabe1960information} emerged in 1960 as an equivalent notion for mutual information in multivariate systems, and yet, it provides negative values in many common systems, contradicting Shannon's viewpoint of information measures as nonnegative quantities. 


Arguably the simplest multivariate setting in which Shannon's mutual information fails to capture the full complexity of the system is that of a three variable system, with \textit{two source variables}, and \textit{one target variable}. Mutual information between the source variables and the target variable does not provide insights about \textit{how} the source variables influence the target variable. Specifically, in various points of the probability space the value of the target variable might be computable either:
\begin{enumerate}
    \item[(a)] exclusively from one source variable (but not the other);
    \item[(b)] either one of the source variables; or
    \item[(c)] both variables jointly (but not separately).
\end{enumerate}
In~2010 William and Beer~\cite{williams2010nonnegative} proposed to formalize the above fine-grained interactions in a three variable system\footnote{William and Beer formulated their notions for any number of source variables, and yet herein we focus on three variables for simplicity. Extending our methods to multiple source variables will be addressed in future version of this paper.} using an axiomatic approach they called \textit{Partial Information Decomposition} (PID). They proposed decomposing said mutual information to four constituent ingredients called information atoms, which capture the above possible interactions between the variables:
\begin{enumerate}
    \item [(a)] two \textit{unique information atoms}, one for each source variable, which capture the information each source variable implies about the target variable, that cannot be inferred from the other;
    \item [(b)] one \textit{redundant information atom}, which captures the information that can inferred about the target variable from either one of the source variables; and
    \item [(c)] one \textit{synergistic information atom}, which captures the information that can be inferred about the target variable from both source variables jointly, but not individually.
\end{enumerate}


Ref.~\cite{williams2010nonnegative} proposed a set of axioms that the above information atoms should satisfy in order to provide said insights, and follow-up works in the field identified several additional properties \cite{ince2017partial,mediano2019beyond,lyu2023system,varley2023generalized}. Yet, in spite of extensive efforts~\cite{griffith2014intersection, ince2017measuring, bertschinger2013shared, harder2013bivariate, bertschinger2014quantifying}, a comprehensive definition of information atoms which satisfies all these axioms and properties is yet to be found.

In spite of limited understanding of the information atoms, PID has already found multiple applications in various fields. As a simple example~\cite[Fig.~1]{banerjee2018computing}, one can imagine the two source variables being education level and gender, and the target variable being annual income. An exact formula for computing the information atoms would shed insightful information about the extent to which annual income is a result of education level, gender, either one, or both.

Beyond this simple example, PID has broad applications in a wide range of fields.
In brain network analysis, PID (or similar ideas) has been instrumental in measuring correlations between neurons~\cite{schneidman2003synergy} and understanding complex neuronal interactions in cognitive processes \cite{varley2023partial}. For privacy and fairness studies, the synergistic concept provides insights about data disclosure mechanisms~\cite{rassouli2019data,hamman2023demystifying}. In the field of causality, information decomposition can be used to distinguish and quantify the occurrence of causal emergence \cite{rosas2020reconciling}, and more.



In this paper, we propose an explicit PID formula that satisfies all of Williams and Beer's axioms, as well as several additional desired properties.
We do so by introducing the \textit{do-operation}, which is inspired by similar concepts in the field of causal analysis~\cite{pearl1995causal,pearl2009causality,hoel2013quantifying}. Intuitively, based on the understanding that unique information is ``ideal conditional mutual information,'' our method first adjusts the entire probability distribution by using the do-operation in order to make the target variable identical to its conditional distribution given one source variable, and then calculates the expectation of mutual information between it and the other source variable under different conditions. It is worth noting that our method is not based on any of the point-wise, localized, or optimization approaches that existing methods use.


We begin in Section~\ref{section:framework} by introducing the PID framework, and its axioms and properties. We continue in Section~\ref{Section 3} by introducing our do-operation and the definition of unique information, from which all other definitions follow, and prove that all axioms and properties are satisfied. We discuss the intrinsic meaning of our definition in Section~\ref{section: discussion}, 
and provide all proofs in the appendix.

\section{Framework, axioms, and properties}\label{section:framework}

The following notational conventions are observed throughout this article:
$X,\mathcal{X},x$ (similarly~$Y,\mathcal{Y},y$ etc.) denote a random variable, its corresponding (finite) alphabet, and an element of that alphabet, respectively. The distribution of~$X$ is denoted by $\mathcal{D}_X$, the joint distribution of~$X$ and $Y$ is denoted by~$\mathcal{D}_{X,Y}$, and the distribution of $X$ given $Y=y$ is denoted by $\mathcal{D}_{X|Y=y}$. 

For random variables~$X,Y,Z$, the quantity~$I((X,Y);Z)$ captures the amount of information that one \textit{target variable}~$Z$ shares with the \textit{source variables}~$(X,Y)$, but provides no further information regarding finer interactions between the three variables. To gain more subtle insights into the interactions between~$Z$ and~$(X,Y)$, \cite{williams2010nonnegative} proposed to further decompose~$I((X,Y);Z)$ into \textit{information atoms}. Specifically, the shared information between~$Z$ and~$(X,Y)$ should contain a \textit{redundant} information atom, two \textit{unique} information atoms, and one~\textit{synergistic} information atom (see Figure~\ref{fig:PID}).

The redundant information atom $\operatorname{Red}(X,Y\to Z)$ (also called ``shared'') represents the information which either~$X$ or~$Y$ imply about~$Z$. The unique information atom $\operatorname{Un}(X \to Z|Y)$ represents the information individually contributed to~$Z$ by~$X$, but not by~$Y$ (similarly $\operatorname{Un}(Y \to Z|X)$). The synergistic information atom $\operatorname{Syn}(X,Y\to Z)$ (also called ``complementary''), represents the information that can only be known about~$Z$ through the \textit{joint} observation of~$X$ and~$Y$, but cannot be provided by either one of them separately. Together, we must have that
\begin{align}
    I((X,Y);Z) &= \operatorname{Red}(X,Y \to Z) +\operatorname{Syn}(X,Y \to Z) \nonumber \\
    &\phantom{=}+\operatorname{Un}(X\to Z\vert Y)+ \operatorname{Un}(Y\to Z\vert X).\label{eqn:basicRelation}
\end{align}
We refer to~\eqref{eqn:basicRelation} as \textit{Partial Information Decomposition} (PID).

Moreover, since the redundant atom together with one of the unique atoms constitute all information that one source variable implies about the target variable, it must be the case that their summation equals the mutual information between the two, i.e., that
\begin{align}
    I(X;Z) &= \operatorname{Red}(X,Y \to Z) + \operatorname{Un}(X\to Z\vert Y), \mbox{ and}\nonumber\\
    I(Y;Z) &= \operatorname{Red}(X,Y \to Z) + \operatorname{Un}(Y\to Z\vert X).\label{equ:Information Atoms' relationship_2}
\end{align}

In a similar spirit, the synergistic information atom and one of the unique information atoms measure shared information between the target variable and one of the source variables, while excluding the other source variable. Therefore, the summation of these quantities should coincide with the well-known definition of conditional mutual information, i.e.,
\begin{align}
    I(Z;X|Y) &= \operatorname{Syn}(X,Y \to Z) + \operatorname{Un}(X\to Z\vert Y),\mbox{ and}\nonumber\\
    I(Z;Y|X) &= \operatorname{Syn}(X,Y \to Z) + \operatorname{Un}(Y\to Z\vert X).\label{equ:Information Atoms' relationship_3}
\end{align}

Eqs.~\eqref{equ:Information Atoms' relationship_2}, and~\eqref{equ:Information Atoms' relationship_3} are the foundation of an axiomatic approach  towards an operational definition of the information atoms. These equations form the first in a series of axioms, presented next, which were raised in previous works on the topic \cite{williams2010nonnegative,williams2011information,griffith2014quantifying}. Such axiomatic approach was also taken in the past in order to shed light on Shannon's mutual information~\cite{csiszar2008axiomatic}.

\begin{axiom}[Information atoms relationship] 
\label{def: Information Atoms' relationship}
Partial Information Decomposition~\eqref{eqn:basicRelation} satisfies~\eqref{equ:Information Atoms' relationship_2} and~\eqref{equ:Information Atoms' relationship_3}.
\end{axiom}

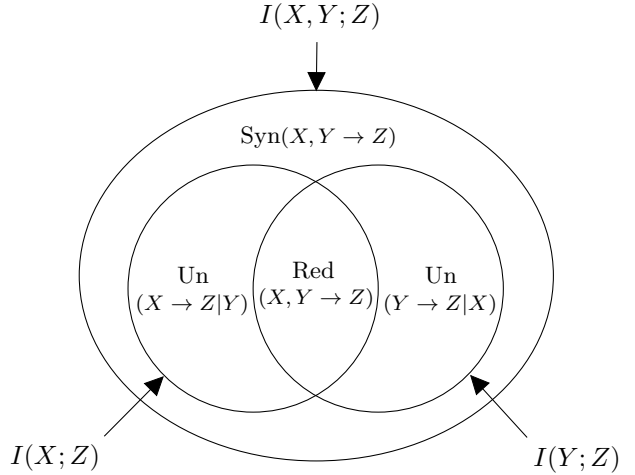
\begin{figure}[t]
\centering

\begin{tikzpicture}[x=0.75pt,y=0.75pt,yscale=-1,xscale=1]

\draw   (149.93,160.59) .. controls (149.93,108.92) and (203.46,67.04) .. (269.5,67.04) .. controls (335.54,67.04) and (389.07,108.92) .. (389.07,160.59) .. controls (389.07,212.25) and (335.54,254.13) .. (269.5,254.13) .. controls (203.46,254.13) and (149.93,212.25) .. (149.93,160.59) -- cycle ;
\draw   (174.54,167.13) .. controls (174.54,132.69) and (202.79,104.77) .. (237.64,104.77) .. controls (272.48,104.77) and (300.73,132.69) .. (300.73,167.13) .. controls (300.73,201.58) and (272.48,229.5) .. (237.64,229.5) .. controls (202.79,229.5) and (174.54,201.58) .. (174.54,167.13) -- cycle ;
\draw    (190.09,213.85) -- (162.55,241.66) ;
\draw [shift={(192.21,211.72)}, rotate = 134.73] [fill={rgb, 255:red, 0; green, 0; blue, 0 }  ][line width=0.08]  [draw opacity=0] (8.93,-4.29) -- (0,0) -- (8.93,4.29) -- cycle    ;
\draw    (349.45,213.32) -- (375.82,242.28) ;
\draw [shift={(347.43,211.1)}, rotate = 47.68] [fill={rgb, 255:red, 0; green, 0; blue, 0 }  ][line width=0.08]  [draw opacity=0] (8.93,-4.29) -- (0,0) -- (8.93,4.29) -- cycle    ;
\draw    (269.54,64.04) -- (269.82,42.72) ;
\draw [shift={(269.5,67.04)}, rotate = 270.74] [fill={rgb, 255:red, 0; green, 0; blue, 0 }  ][line width=0.08]  [draw opacity=0] (8.93,-4.29) -- (0,0) -- (8.93,4.29) -- cycle    ;
\draw   (237.64,167.13) .. controls (237.64,132.69) and (265.89,104.77) .. (300.73,104.77) .. controls (335.58,104.77) and (363.83,132.69) .. (363.83,167.13) .. controls (363.83,201.58) and (335.58,229.5) .. (300.73,229.5) .. controls (265.89,229.5) and (237.64,201.58) .. (237.64,167.13) -- cycle ;

\draw (235,152) node [anchor=north west][inner sep=0.75pt]  [font=\small] [align=left] { \ \ \ \ \ $\operatorname{Red}$\\{\footnotesize \text{ }$(X,Y\to Z)$}};
\draw (177,155) node [anchor=north west][inner sep=0.75pt]  [font=\small] [align=left] { \ \ \ \ \ $ \operatorname{Un}$\\{\footnotesize $(X\to Z|Y)$}};
\draw (230,84) node [anchor=north west][inner sep=0.75pt]  [font=\small] [align=left] {$\operatorname{Syn}${\footnotesize $(X,Y\to Z)$}};
\draw (302,155) node [anchor=north west][inner sep=0.75pt]  [font=\small] [align=left] { \ \ \ \ \ $\operatorname{Un}$\\{\footnotesize $(Y\to Z|X)$}};
\draw (239.45,22.05) node [anchor=north west][inner sep=0.75pt]  [font=\normalsize] [align=left] {$I(X,Y ; Z)$};
\draw (377.82,245.28) node [anchor=north west][inner sep=0.75pt]   [align=left] {$I(Y ; Z)$};
\draw (113.77,244.33) node [anchor=north west][inner sep=0.75pt]   [align=left] {$I(X ; Z)$};
\end{tikzpicture}
\caption{A pictorial representation of Partial Information Decomposition~\eqref{eqn:basicRelation}, where~$I((X,Y);Z)$ is decomposed to its finer information atoms, the synergistic $\operatorname{Syn}(X,Y\to Z)$ (also called ``complementary''), the redundant $\operatorname{Red}(X,Y\to Z)$ (also called ``shared''), and the two directional unique components~$\operatorname{Un}(X\to Z|Y)$ and~$\operatorname{Un}(Y\to Z|X)$. The summation of the redundant atom and one of the unique atoms must be equal to the corresponding mutual information, as described in Eq.~\eqref{equ:Information Atoms' relationship_2}.}
\label{fig:PID}
\end{figure}

Notice that it suffices to specify the definition of any one of the information atoms, and the definitions for the remaining atoms follow from Axiom~\ref{def: Information Atoms' relationship}. Consequently,~\cite{williams2010nonnegative,lizier2013towards} chose to specify $\operatorname{Red}$, and provided three additional axioms which~$\operatorname{Red}$ should satisfy. 

The first additional axiom is \textit{commutativity} of the source variables, which implies that the order of the source variables must not affect the value of the redundant information. 

\begin{axiom} [Commutativity]
\label{Axiom: Commutativity}
Partial Information Decomposition satisfies $\operatorname{Red}(X,Y\to Z) = \operatorname{Red}(Y,X\to Z)$.
\end{axiom}   

The second is \textit{monotonicity}, which implies that the redundant information is non-increasing when adding a source variable, since the newly added variable cannot increase the redundancy between the original variables. We sidestep the discussion about monotonicity with more than two variables, which is not our focus in this paper, even though it can be easily obtained by extending our definition to more than two source variables.

The third is \textit{self-redundancy}, which defines the redundant information from \textit{one} source variable to the target variable (i.e.,~$\operatorname{Red}(X\to Z)$) as the mutual information between them. In the case of two source variables considered herein, monotonicity and self-redundancy merge into the following single axiom.

\begin{axiom} [Monotonicity and self-redundancy]
\label{axiom: Monotonicity and Self-redundancy}
Partial Information Decomposition satisfies $\operatorname{Red}(X,Y \to Z) \le \min\{I(X;Z),I(Y;Z)\}$.
\end{axiom}

Notice that Axiom~\ref{axiom: Monotonicity and Self-redundancy}, alongside Axiom~\ref{def: Information Atoms' relationship} (specifically~\eqref{equ:Information Atoms' relationship_2}), imply that~$\operatorname{Un}$ is a nonnegative quantity. The nonnegativity of~$\operatorname{Red}$ is stated in~\cite{williams2010nonnegative,lizier2013towards} as a separate axiom, shown next.

\begin{axiom} [Nonnegativity]
\label{Axiom: Nonnegativity}
Partial Information Decomposition satisfies $\operatorname{Red}(X,Y \to Z) \ge 0$.
\end{axiom}

The nonnegativity of~$\operatorname{Syn}$ is normally not listed as an axiom, since it is debatable if it should or should not be nonnegative; we will show that our method yields nonnegative~$\operatorname{Syn}$ under the \textit{closed-system assumption} (i.e.,~$H(Z|X,Y)=0$) in Section~\ref{Section 3}, and further discussion is given in Section~\ref{section: discussion}.

Besides, subsequent to~\cite{williams2010nonnegative,lizier2013towards}, studies suggested two additional properties, \textit{additivity} and \textit{continuity}~\cite{bertschinger2014quantifying,rauh2023continuity}. Additivity implies that whenever independent variable systems are considered, the joint information measures should be the sum of the information measures of each individual system. This is the case, for instance, in joint entropy of two independent variables. 

\begin{property} [Additivity] 
\label{property: Additivity}
Partial Information Decomposition of two independent systems
$\mathcal{D}_{X,Y,Z}$ and $\mathcal{D}_{\bar{X},\bar{Y},\bar{Z}}$ satisfy 
\begin{align*}
\operatorname{Un}&((X,\bar{X}) \to (Z,\bar{Z}) | (Y,\bar{Y}))\nonumber \\
&= \operatorname{Un}(X\to Z|Y)+ \operatorname{Un}(\bar{X}\to \bar{Z}|\bar{Y}), \mbox{ and}\\
\operatorname{F}&((X,\bar{X}),(Y,\bar{Y})\to (Z,\bar{Z}))=\\
&\operatorname{F}(X,Y\to Z) + \operatorname{F}(\bar{X},\bar{Y}\to \bar{Z}),
\end{align*}
for every~$\operatorname{F}\in\{\operatorname{Red},\operatorname{Syn}\}$.
\end{property}

Continuity implies that small changes in the probability distribution lead to small changes in the value of the information measure. It ensures that the measure behaves predictably and is a key property in information theory, particularly for measures like entropy and mutual information.

\begin{property} [Continuity] 
\label{property: Continuity}
$\operatorname{Red}$, $\operatorname{Un}$, and~$\operatorname{Syn}$ are continuous functions from the underlying joint distributions $\mathcal{D}_{X,Y,Z}$ to $\mathbb{R}$.
\end{property}

In addition, another well-known property is \textit{independent identity}~\cite{ince2017measuring}, which asserts that in a system of two independent source variables and a target variable which equals to their joint distribution, the redundant information should be zero.

\begin{property}[Independent Identity]
\label{property: Independent Identity}
If $I(X,Y)=0$ and $Z=(X,Y)$, then $\operatorname{Red}(X,Y\to Z) = 0$.
\end{property}

We mention that several important properties can be inferred from the above. For example, the non-negativity of~$\operatorname{Un}$ can be obtained from Axiom~\ref{def: Information Atoms' relationship} and Axiom~\ref{axiom: Monotonicity and Self-redundancy} as mentioned earlier; the commutativity of~$\operatorname{Syn}$ follows from Axiom~\ref{def: Information Atoms' relationship} and Axiom~\ref{Axiom: Commutativity}; the difference between~\eqref{equ:Information Atoms' relationship_2} and~\eqref{equ:Information Atoms' relationship_3} is often called \textit{consistency}~\cite{bertschinger2014quantifying}, etc.

Finally, we emphasize once again that none of the existing operational definitions of the information atoms satisfy all of the above. A comprehensive list of violations is beyond the page limit of this paper, and yet we briefly mention that Axiom~\ref{Axiom: Nonnegativity} (nonnegativity) is violated by \cite{griffith2014intersection,ince2017measuring,finn2018pointwise} (although some sources do not refer to non-negativity as a requirement); Property~\ref{property: Additivity} (additivity) is violated by all works except~\cite{bertschinger2014quantifying}, ~\cite{griffith2015quantifying}, ~\cite{griffith2014intersection}, and~\cite{kolchinsky2022novel} according to~\cite{rauh2023continuity}; Property~\ref{property: Independent Identity} (independent identity) is violated by~\cite{williams2010nonnegative}; Property~\ref{property: Continuity} (continuity) is violated by \cite{harder2013bivariate,griffith2014intersection}, \cite{griffith2015quantifying}, \cite{kolchinsky2022novel}, etc.

\section{Proposed information decomposition definition}
\label{Section 3}
In this section, we present our operational definition of~$\operatorname{Un}$, from which the definitions of the remaining information atoms follow. 
Then, we explain the logic behind this definition, and prove that it satisfies all the axioms and properties proposed in~Section~\ref{section:framework}. 

\subsection{Definition of Information Atoms}
Given a system~$X,Y,Z$ define a new random variable~$X'$ over the alphabet~$\mathcal{X}$ via its conditional joint distribution as follows:
\begin{align}\label{equation:X'Z}
&\Pr(X'=x,Z=z|Y=y) \nonumber \\
&\triangleq\Pr(X=x|Z=z)\Pr(Z=z|Y=y),
\end{align}
meaning, for every~$x\in\mathcal{X}$,
\begin{align*}
    \Pr(X'=x)=\sum_{y,z\in\mathcal{Y}\times \mathcal{Z}}\Pr(X'=x,Z=z|Y=y)\Pr(Y=y).
\end{align*}
The variable~$X'$ is well-defined since all probabilities are non-negative, and since
\begin{align*}
    &\sum_{x\in\mathcal{X}}\Pr(X'=x)\nonumber \\ &=\sum_{x,y,z\in\mathcal{X}\times\mathcal{Y}\times \mathcal{Z}}\Pr(X'=x,Z=z|Y=y)\Pr(Y=y) \nonumber \\ 
    &=\sum_{y\in\mathcal{Y}}\Pr(Y=y) \nonumber =1.
\end{align*}
\begin{definition}[Unique Information]\label{definition:un}
The unique information from~$X$ to~$Z$ given~$Y$ is $\operatorname{Un}(X \to Z | Y) = I(X';Z|Y)$.
\end{definition}

The definitions for the remaining information atoms are then implied by Axiom \ref{def: Information Atoms' relationship} as follows. 

\begin{definition}[Redundant Information] 
\label{definition:red}
The Redundant Information from $X$ and $Y$ to $Z$ is defined as:
\begin{align*}
\operatorname{Red}(X,Y\to Z) = I(X;Z) - \operatorname{Un}(X \to Z | Y).
\end{align*}
\end{definition}

\begin{definition}[Synergistic Information] 
\label{definition:syn}
The synergistic information from $X$ and $Y$ to $Z$ is defined as:
\begin{align*}
\operatorname{Syn}(X,Y\to Z) = I(X;Z|Y) - \operatorname{Un}(X \to Z | Y).
\end{align*}
\end{definition}

It should be noted that Definition~\ref{definition:red} and Definition~\ref{definition:syn} strictly depend on the order of the source variables; the commutativity of~$\operatorname{Red}$ (Axiom~\ref{Axiom: Commutativity}) will be addressed in the sequel, and the commutativity of~$\operatorname{Syn}$ follows from Axiom~\ref{def: Information Atoms' relationship} and Axiom~\ref{Axiom: Commutativity} as mentioned earlier. 
\subsection{Intuitive Explanation of Definition~\ref{definition:un}.}

Our definition of unique information~$\operatorname{Un}$ is derived from a newly defined \textit{do-operation}. 

\begin{definition} [Do-operation] 
\label{def:do operation}
Given $\mathcal{D}_{X,Z}$ and $\mathcal{D}_C $ such that the support of $\mathcal{D}_C $ is contained in the support of~$\mathcal{D}_Z$, let $do(\mathcal{D}_{X,Z}|\mathcal{D}_{C})=\mathcal{D}_{A,C}$, where
\begin{align}
\label{equ:def of do operation}
\Pr(A=x,C=z) =
\Pr(X=x|Z=z) \Pr(C=z)
\end{align}
for all $x,z \in \mathcal{X}\times \mathcal{Z}$.
\end{definition}

\ifthenelse{\boolean{IncludeAppendix}}{In Lemma~\ref{lem:support of definition do}, which is given and proved in Appendix~\ref{proof:def do operation}, it is shown }
{In \cite{lyu2024explicit}, we show }
that $\mathcal{D}_{A,C}$ in Definition~\ref{def:do operation} is well-defined in the sense that the right marginal of~$\mathcal{D}_{A,C}$ is identical to~$\mathcal{D}_C$. Therefore, there is no ambiguity in referring to both the input distribution and the right marginal of the output distribution by the same letter~$C$. This operation receives~$\mathcal{D}_{X,Z}$ and~$\mathcal{D}_C$, and outputs a joint distribution~$\mathcal{D}_{A,C}$ whose right marginal is~$\mathcal{D}_C$, and~$\mathcal{D}_{A|C=z}=\mathcal{D}_{X|Z=z}$ for all~$z\in\mathcal{Z}$. Using the do-operation,~$\operatorname{Un}$ can be defined equivalently as follows:
\begin{definition}[Unique Information, equivalent definition]\label{definition:un2}
For~$y\in\mathcal{Y}$ let~$C_y$ be a random variable with distribution~$\mathcal{D}_{C_y}=\mathcal{D}_{Z|Y=y}$, and let~$\mathcal{D}_{A_y,C_y}=do(\mathcal{D}_{X,Z}|\mathcal{D}_{C_y})$. The unique information from~$X$ to~$Z$ given~$Y$ is defined as:
\begin{align}
\operatorname{Un}(X \to Z | Y) 
&= \sum_{y\in \mathcal{Y}} \Pr(Y=y) I(A_y;C_y). \nonumber 
\end{align}
\end{definition}
The proof of equivalence is simple, and is given in 
\ifthenelse{\boolean{IncludeAppendix}}{Appendix~\ref{proof:equavelence of definitions}.}
{\cite{lyu2024explicit}.}
The idea behind Definition~\ref{definition:un2} is that by setting~$C=(Z|Y=y)$ for some~$y\in\mathcal{Y}$ in Definition~\ref{def:do operation}, and then by averaging the resulting mutual information values over all~$y\in\mathcal{Y}$, we eliminate the effect of~$Y$ from the directional dependence between~$X$ and~$Z$, without changing the directional dependence itself. 
A similar idea exists in Judea Pearl's do-calculus~\cite{pearl2009causality} (also \cite{hoel2013quantifying}), where a node in a Bayesian network is set to a certain value while removing all incoming dependencies to that node, thereby distilling the causal relationship between that value and the remainder of the network.

\subsection{Satisfaction of axioms and properties}


To show that our definition satisfies the axioms and properties mentioned in Section~\ref{section:framework}, we require the following technical lemma. The proof is given in 
\ifthenelse{\boolean{IncludeAppendix}}{Appendix~\ref{proof:def conditional entropy}.}
{\cite{lyu2024explicit}.}

\begin{lemma}\label{le:invariant property of channel}
Following the notations of Definition~\ref{definition:un}, we have that $H(X|Z) = H(X'|Z)=H(X'|Z,Y)$, and that $H(X')=H(X)$.
\end{lemma}
\begin{corollary}
\label{corollary:sumH-H=U}
Unique information (Def.~\ref{definition:un}) can be written as:
\begin{align*}
\operatorname{Un}(X \to Z | Y)&= I(X';Z|Y)=X(X'|Y)-H(X'|Z,Y)\\
&=H(X'|Y)-H(X|Z).
\end{align*}    
\end{corollary}
\begin{corollary}
\label{cor:sumH(A)<H(X)}
By Lemma~\ref{le:invariant property of channel}, and since conditioning reduces entropy, we have that $H(X'|Y)\le H(X')=H(X)$.
\end{corollary}

Based on the above lemmas and corollaries, we are in a position to prove that our definition of~$\operatorname{Un}$ satisfies the required axioms.

\subsubsection{Proof of Axiom \ref{def: Information Atoms' relationship}, Information atoms relationship}
Follows immediately from Definition~\ref{definition:red} and Definition~\ref{definition:syn}.

\subsubsection{Proof of Axiom \ref{Axiom: Commutativity}, Commutativity}
First, Definition \ref{definition:un} and  Definition \ref{definition:red} provide the following equivalent way for computing~$\operatorname{Red}$, which is proved in 
\ifthenelse{\boolean{IncludeAppendix}}{Appendix~\ref{proof: redefinition of redundant information}.}
{\cite{lyu2024explicit}.}
\begin{lemma} \label{le:redefinition of redundant information}
Redundant information (Def.~\ref{definition:red})  can alternatively 
be written as $\operatorname{Red}(X,Y\to Z)=I(X';Y)$.
\end{lemma}
Similarly, by switching between~$X$ and~$Y$ in Definition \ref{definition:red} we have that $\operatorname{Red}(Y, X \to Z) = I(Y;Z) - \operatorname{Un}(Y \to Z | X)$; based on Lemma~\ref{le:redefinition of redundant information}, this equals to $I(X;Y')$, where~$Y'$ is defined analogously to~$X'$ in~\eqref{equation:X'Z}, i.e., 
\begin{align}\label{equation:Y'Z}
&\Pr(Y'=y,Z=z|X=x) \nonumber\\
&=\Pr(Y=y|Z=z)\Pr(Z=z|X=x).
\end{align}

Then, we can conclude the commutativity of redundant information through the following lemma, which is proved in
\ifthenelse{\boolean{IncludeAppendix}}{Appendix~\ref{proof:Commutativity of Redundant Information}.}
{\cite{lyu2024explicit}.}

\begin{lemma}[Commutativity of Redundant Information]
\label{lem: Commutativity} 
For~$X'$ and~$Y'$ as above, we have that $I(X';Y) = I(X;Y')$.
\end{lemma}
Combining Lemma~\ref{le:redefinition of redundant information} and Lemma~\ref{lem: Commutativity} readily implies the commutativity of~$\operatorname{Red}$, i.e., $\operatorname{Red}(X,Y \to Z) = \operatorname{Red}(Y,X \to Z).$

\subsubsection{Proof of Axiom \ref{axiom: Monotonicity and Self-redundancy}, Monotonicity and self-redundancy}
According to Definition~\ref{definition:un}, $\operatorname{Un}$ is nonnegative as conditional mutual information. Therefore, $\operatorname{Red}(X,Y\to Z) \le I(X;Z)$ by Definition~\ref{definition:red}. Similarly, by Definition~\ref{definition:red} we have~$\operatorname{Red}(Y,X\to Z)=I(Y;Z)-\operatorname{Un}(Y\to Z|X)$, and hence $\operatorname{Red}(Y,X\to Z)\le I(Y;Z)$. Since~$\operatorname{Red}$ is symmetric by Axiom~\ref{Axiom: Commutativity}, it follows that
\begin{align*}
    \operatorname{Red}(X,Y\to Z)\le \min\{I(X;Z),I(Y;Z)\}.
\end{align*}




\subsubsection{Proof of Axiom \ref{Axiom: Nonnegativity}, Nonnegativity}
We begin by showing that~$\operatorname{Red}$ is nonnegative, for which we require the following lemma, proved in
\ifthenelse{\boolean{IncludeAppendix}}{Appendix~\ref{proof:un<I}.}
{\cite{lyu2024explicit}.}
\begin{lemma}
\label{le:un<I}
Unique information (Def.~\ref{definition:un}) is bounded from above by mutual information, i.e.,
    \begin{align*}
        \operatorname{Un}(X \to Z | Y) \le I (X;Z).
    \end{align*}
\end{lemma}

Then, nonnegativity of~$\operatorname{Red}$ follows from Lemma~\ref{le:un<I}.

\begin{corollary}[Nonnegativity of Redundant Information]
Redundant information (Def.~\ref{definition:red}) is nonnegative, i.e.,
\begin{align*}
\operatorname{Red}(X,Y\to Z) \ge 0.
\end{align*}
\end{corollary}

Finally, we remark that even though it is not a required axiom, non-negativity of~$\operatorname{Syn}$ can be proved under an additional assumption as follows.

\begin{remark}[Nonnegativity of Synergistic Information]
\label{remark:Nonnegativity of Synergistic Information}
    Suppose~$H(Z|X,Y)=0$ (closed system assumption). Since $\operatorname{Un}(X \to Z | Y)=I(X';Z|Y)\le H(Z|Y)$, it follows from Definition~\ref{definition:syn} that $\operatorname{Syn}(X,Y\to Z) \ge 0$.
\end{remark}
    




\subsubsection{Proof of Property \ref{property: Additivity}, Additivity}
The following lemma is proved in
\ifthenelse{\boolean{IncludeAppendix}}{Appendix~\ref{proof:Additivity of Unique Information}.}
{\cite{lyu2024explicit}.}
\begin{lemma}[Additivity of Unique Information]
\label{le:Additivity of Unique Information}
For two independent sets of variables $X,Y,Z$ and $\bar{X},\bar{Y},\bar{Z}$, unique information (Def.~\ref{definition:un}) is additive:
\begin{align}
&\operatorname{Un}((X,\bar{X}) \to (Z,\bar{Z}) | (Y,\bar{Y}))\nonumber \\
&= \operatorname{Un}(X\to Z|Y)+ \operatorname{Un}(\bar{X}\to \bar{Z}|\bar{Y})
\end{align}
\end{lemma}

Since mutual information and conditional entropy are additive in the above sense, by Definition~\ref{definition:red} and Definition~\ref{definition:syn}, alongside Lemma~\ref{le:Additivity of Unique Information}, $\operatorname{Red}$ and~$\operatorname{Syn}$ are additive as well.

\subsubsection{Proof of Property \ref{property: Continuity}, Continuity}
We begin by showing that~$\operatorname{Red}$ is continuous, for which we require the following lemma proved in
\ifthenelse{\boolean{IncludeAppendix}}{Appendix~\ref{proof:continuity of do operation}.}
{\cite{lyu2024explicit}.}
\begin{lemma}[Continuity of Redundant Information]
\label{le:continuity of do operation}
The redundant information (Def.~\ref{definition:red}) is a continuous function of the input distribution~$\mathcal{D}_{X,Y,Z}$ to $\mathbb{R}$.
\end{lemma}

By Definition~\ref{definition:red} and Definition~\ref{definition:syn}, the continuity of~$\operatorname{Un}$ and $\operatorname{Syn}$ can also be derived.

\subsubsection{Proof of Property \ref{property: Independent Identity}, Independent Identity}
The following lemma is proved in
\ifthenelse{\boolean{IncludeAppendix}}{Appendix~\ref{proof:Independent Identity}.}
{\cite{lyu2024explicit}.}

\begin{lemma}
\label{le:Independent Identity}
The operator~$\operatorname{Red}$ satisfies Property~\ref{property: Independent Identity}.
\end{lemma}

\section{Discussion}
\label{section: discussion}
In this paper, we proposed an explicit operational formula for PID, which is distinct from any existing approach, and proved that it satisfies all axioms and properties. In this section we provide further intuitive explanation for our approach.

First, we wish to elucidate the role that our do-operation plays in the definition of~$\operatorname{Un}$ (Definition~\ref{definition:un2}).
In a sense, the do-operation can be understood as adjusting the marginal distribution of the~$Z$ variable of~$\mathcal{D}_{X,Z}$, while impacting its connections with other variables as little as possible.
This understanding can be confirmed by Lemma~\ref{le:invariant property of channel}, which shows that the expected value of the conditional entropy after the do-operation retains its original value.
This resembles the invariance implied in Shannon's communication model~\cite{shannon2001mathematical}, where the conditional entropy of the output given the input is not affected by the input distribution.
From this perspective, $Z$ and $X$ can be regarded as the input and output of the channel, respectively, that indicates their ``relationship.'' The do-operation changes the distribution of the input~$Z$, but does not change the channel's characteristic (i.e.,~$H(X|Z)$). 

Based on this, Definition \ref{definition:un2} realizes the intuition that unique information should represent the relationship between source variable and target variable given the other source variable. So, we use the do-operation to control the marginal distribution of the target variable $Z$ to its conditional distribution given the value $y$ of some source variable~$Y$, and then use the expectation of mutual information $\sum_{y} \Pr(Y=y)I(A_y;C_y)$ to capture the ``connection'' between the specific source variable $X$ and target variable $Z$ given $Y$ after the do-operation. 

The reason this method can partition $I(X;Z|Y)$ to~$\operatorname{Syn}$ and~$\operatorname{Un}$ (Def.~\ref{definition:syn}), 
is that the do-operation eliminates high-order relations between~$Y$ and~$X,Z$, i.e.~$\operatorname{Syn}$. Specifically, conditional mutual information relies on the joint conditional probability $\mathcal{D}_{(X,Z)|y}$, in expectation over all~$y\in\mathcal{Y}$. This distribution includes both the conditional influence of $Y$ on $X,Z$, but also has a simultaneous influence on the relationship between~$X$ and~$Z$. 

However, Definition \ref{definition:un2} of unique information retains the relationship between $X$ and $Z$ without influence from $Y$ by using the conditional probability $\mathcal{D}_{Z|y}$, in expectation over all~$y\in\mathcal{Y}$, to perform the do-operation, which only reflects the conditional influence of $Y$ on $Z$. Therefore, the expectation of mutual information $\sum_{y} \Pr(Y=y)I(A_y;C_y)$ can accurately quantify the unique information, which represents the pure conditional mutual relationship.

In addition to the above analysis of do-operations in unique information, Lemma \ref{le:redefinition of redundant information} also brings another perspective worth discussing. Redundant information can be understood as the mutual information $I(X';Y)$ (or $I(X;Y')$) obtained by changing the joint probability distribution $\mathcal{D}_{X,Y}$ according to $\mathcal{D}_{X,Y,Z}$ without changing the marginal distribution $\mathcal{D}_{X}$ (or $\mathcal{D}_{Y}$) according to Lemma \ref{le:invariant property of channel} ($H(X')=H(X)$). 

As mentioned earlier, our definition of~$\operatorname{Syn}$ might be negative, unless the system is closed (i.e., $H(Z|X,Y)=0$, Remark~\ref{remark:Nonnegativity of Synergistic Information}). 
While~$\operatorname{Un}$ and~$\operatorname{Red}$ represent the information shared by one or two source variables with the target variable, $\operatorname{Syn}$ represents the information provided to the target variable by the ``cooperation'' of source variables. It is an accepted aphorism 
that cooperation does not necessarily increase outcome, and hence it might be the case that negative values of~$\operatorname{Syn}$ conform with intuition. However, the reason why this explanation is no longer necessary in a closed system, as well as alternative interpretations of~$\operatorname{Syn}$ that are nonnegative, remain to be studied.

\noindent\textbf{Acknowledgments.} The authors would like to thank an anonymous reviewer whose suggestions greatly simplified the paper. This research was supported by AFOSR grants FA9550-22-1-0054 and FA9550-23-1-0208.

\IEEEtriggeratref{15}


\newpage
\bibliographystyle{unsrt}
\bibliography{references}




\ifthenelse{\boolean{IncludeAppendix}}{

\newpage
\appendix


\subsection{Proof of the completeness of Definition~\ref{def:do operation}.}
In this part, we will show that do-operation's output is a probability distribution with the same marginal distribution as its input.
\begin{lemma}
\label{lem:support of definition do} 
For $\mathcal{D}_{X,Z}$ and $\mathcal{D}_C$ as in Definition~\ref{def:do operation}, the output $\Pr(A, C = x, z)$ of~\eqref{equ:def of do operation} describes a probability distribution, i.e.,
\begin{align*}
0 \le \Pr(A=x,C=z) &\le 1,\mbox{ and}\\
\sum_{x,z\in \mathcal{X}\times\mathcal{Z}}\Pr(A=x,C=z) &= 1.
\end{align*}


Furthermore, the marginal distribution $\mathcal{D}_{C}$ of the output $\mathcal{D}_{A,C}$ is equal to the input (call it $\mathcal{D}_{C'}$ in this section), i.e.,.
\begin{align*}
\sum_{x\in \mathcal{X}} \Pr(A=x,C=z) = \Pr(C'=z).
\end{align*}
\end{lemma}
\label{proof:def do operation}


\begin{proof}
We begin by showing $0 \le \Pr(A,C=(x,z)) \le 1$. By Definition \ref{def:do operation},
\begin{align}
\label{equ:Support of Definition_1}
&\Pr(A=x,C=z) \nonumber \\
&= \Pr(X=x|Z=z) \Pr(C'=z).
\end{align}

Since both terms in~\eqref{equ:Support of Definition_1} are between~$0$ and~$1$, so is~$\Pr(A=x,C=z)$.

We continue by showing that
\begin{align}
\label{equ:Support of Definition_3}
&\sum_{x,z}\Pr(A=x,C=z) \nonumber\\
&\overset{\text{Def.~\ref{def:do operation}}}{=} \sum_{x,z}\Pr(X=x|Z=z) \Pr(C'=z) \nonumber\\
&= \sum_{z\in \mathcal{Z}}\Pr(C'=z)\sum_{x\in \mathcal{X}}\Pr(X=x|Z=z) \nonumber\\
&= \sum_{z\in\mathcal{Z}}\Pr(C'=z)=1.
\end{align}
In~\eqref{equ:Support of Definition_3}, we can also conclude that $\sum_{x \in \mathcal{X}}\Pr(A=x,C=z)=\Pr(C'=z)$, which implies that the input~$\mathcal{D}_{C'}$ is equal to the marginal distribution $\mathcal{D}_{C}$ of the output~$\mathcal{D}_{A,C}$. 
\end{proof}

\subsection{Proof of the equivalence of Definitions \ref{definition:un} and \ref{definition:un2}.}\label{proof:equavelence of definitions}

\begin{lemma}
Definition~\ref{definition:un} and Definition~\ref{definition:un2} are equivalent.
\end{lemma}
\begin{proof}
For the purpose of the proof, let
\begin{align}
\label{equ:equavelence of definitions_1}
\operatorname{Un}(X \to Z | Y) 
&= \sum_{y\in \mathcal{Y}} \Pr(Y=y) I(A_y;C_y), \text{ and}\\
\operatorname{Un}'(X \to Z | Y) 
&= I(X';Z|Y).
\end{align}
where~$\mathcal{D}_{C_y}=\mathcal{D}_{Z|Y=y}$. By Definition~\ref{def:do operation} we have
$\mathcal{D}_{A_y,C_y}=do(\mathcal{D}_{X,Z}|\mathcal{D}_{C_y})$, i.e.,
\begin{align*}
\Pr(A_y=x,C_y=z) &=
\Pr(X=x|Z=z) \Pr(Z=z|Y=y)\\
&\overset{\eqref{equation:X'Z}}{=}\Pr(X'=x,Z=z|Y=y).
\end{align*}
Therefore, we have
\begin{align*}
\eqref{equ:equavelence of definitions_1}&=\sum_{y\in \mathcal{Y}} \Pr(Y=y) I(X',Z|Y=y) \\
&=I(X',Z|Y) = \operatorname{Un}'(X \to Z | Y). \qedhere
\end{align*}
\end{proof}

\subsection{Proof of Lemma \ref{le:invariant property of channel}.}
\label{proof:def conditional entropy}
\begin{proof}
By Definition \ref{definition:un}, for every~$x,y,z\in\mathcal{X}\times\mathcal{Y}\times\mathcal{Z}$, we have 
\begin{align}\label{equ:proof definition:un_1}
&\Pr(X'=x,Z=z|Y=y)=\nonumber\\
&=\Pr(X=x|Z=z)\Pr(Z=z|Y=y),
\end{align}
and therefore
\begin{align}
\frac{\Pr(X'=x,Z=z|Y=y)}{\Pr(Z=z|Y=y)}=\Pr(X=x|Z=z),\nonumber
\end{align}
which can be simplified to
\begin{align}
\Pr(X'=x|Z=z,Y=y)=\Pr(X=x|Z=z).\nonumber
\end{align}
Therefore, we have that $H(X'|Z=z,Y=y)=H(X|Z=z)$ for every~$y,z\in\mathcal{Y}\times\mathcal{Z}$, which implies that $H(X'|Z,Y)=H(X|Z)$.

Also, from~\eqref{equ:proof definition:un_1} we have
\begin{align}
&\Pr(X'=x,Z=z, Y=y) \nonumber \\
&=\Pr(X=x|Z=z)\Pr(Z=z,Y=y),\nonumber 
\end{align}
and summation over all~$y\in\mathcal{Y}$ yields
\begin{align}
\label{equ:proof definition:un_2}
\Pr(X'=x,Z=z) 
&=\Pr(X=x|Z=z)\Pr(Z=z)\nonumber \\
&=\Pr(X=x,Z=z),
\end{align}
which readily implies that $H(X'|Z=z)=H(X|Z=z)$ for every~$z\in\mathcal{Z}$, and hence~$H(X'|Z)=H(X|Z)$. 
Also, by~\eqref{equ:proof definition:un_2}, we have
\begin{align*}
\Pr(X'=x)&=\sum_{z \in \mathcal{Z}}\Pr(X'=x,Z=z) \nonumber \\
&=\sum_{z \in \mathcal{Z}}\Pr(X=x,Z=z) \nonumber \\
&=\Pr(X=x),
\end{align*}
which implies that $H(X')=H(X)$.
\end{proof}

\subsection{Proof of Lemma \ref{le:redefinition of redundant information}.}
\label{proof: redefinition of redundant information}
\begin{proof}
By Definition \ref{definition:red}, we have:
\begin{align}
\label{equ:Commutativity of Redundant Information_1}
\operatorname{Red}&(X,Y\to Z) = I(X,Z) - \operatorname{Un}(X \to Z | Y)\nonumber\\
&=(I(X,Z) + H(X|Z)) - (\operatorname{Un}(X \to Z | Y) + H(X|Z))\nonumber\\
&=H(X) - (\operatorname{Un}(X \to Z | Y) + H(X|Z))
\end{align}
Since Corollary~\ref{corollary:sumH-H=U} states that
\begin{align*}
\operatorname{Un}(X \to Z | Y)
&= H(X'|Y) - H(X|Z),
\end{align*}
it follows that
\begin{align*}
\eqref{equ:Commutativity of Redundant Information_1}&=H(X) - H(X'|Y)\\
&\overset{\text{Lem.}~\ref{le:invariant property of channel}}{=}H(X') - H(X'|Y)=I(X';Y).\qedhere
\end{align*}


\end{proof}

\subsection{Proof of Lemma \ref{lem: Commutativity}.}
\label{proof:Commutativity of Redundant Information}
\begin{proof}
To prove that $I(X';Y) = I(X;Y')$, it suffices to show that $\Pr(X'=x,Y=y,Z=z)=\Pr(X=x,Y'=y,Z=z)$ for all~$x,y,z\in\mathcal{X}\times\mathcal{Y}\times\mathcal{Z}$. 
We have 
\begin{align}
\label{equ: Commutativity of Redundant Information_1}
&\Pr(X'=x,Y=y,Z=z) \nonumber \\
&=\Pr(X'=x,Z=z|Y=y) \Pr(Y=y)\nonumber\\
&\overset{\eqref{equation:X'Z}}{=}
\Pr(X=x|Z=z)\Pr(Z=z|Y=y) \Pr(Y=y)\nonumber\\
&=\Pr(X=x|Z=z)\Pr(Z=z,Y=y) \nonumber \\
&=\Pr(X=x|Z=z)\Pr(Y=y|Z=z)\Pr(Z=z),
\end{align}
in which~$X$ and~$Y$ are symmetric. The proof is concluded by following similar steps in reversed order, i.e.,
\begin{align*}
\eqref{equ: Commutativity of Redundant Information_1}
&=\Pr(X=x,Z=z)\Pr(Y=y|Z=z) \nonumber \\
&=\Pr(Z=z|X=x)\Pr(Y=y|Z=z)\Pr(X=x) \nonumber \\
&\overset{\eqref{equation:Y'Z}}{=}\Pr(Y'=y,Z=z|X=x)\Pr(X=x) \nonumber \\
&=\Pr(X=x,Y'=y,Z=z).\qedhere
\end{align*}
\end{proof}

\subsection{Proof of Lemma \ref{le:un<I}.}
\label{proof:un<I}
\begin{proof}
    By Corollary~\ref{corollary:sumH-H=U}, Lemma~\ref{le:un<I} is equivalent to
\begin{align*}
H(X'|Y) - H(X|Z) &\le I(X,Z)\nonumber\\
 H(X'|Y) &\le H(X),
\end{align*}
which coincides with Corollary~\ref{cor:sumH(A)<H(X)}, and the proof follows.
\end{proof}

\subsection{Proof of Lemma \ref{le:Additivity of Unique Information}.}
\label{proof:Additivity of Unique Information}
\begin{proof}
By Definition \ref{definition:un}, we have:
\begin{align}
\label{equ:proof of superposition_1}
&\operatorname{Un}((X,\bar{X}) \to (Z,\bar{Z}) | (Y,\bar{Y})) \nonumber \\
&=  I((X,\bar{X})' ; (Z,\bar{Z}) | (Y,\bar{Y}))
\end{align}
where 
\begin{align}
\label{equ:proof of superposition_2}
&\Pr((X,\bar{X})',(Z,\bar{Z})=(x,\bar{x}),(z,\bar{z})|(Y,\bar{Y})=(y,\bar{y}))\nonumber \\
&=\Pr((X,\bar{X})=(x,\bar{x})|(Z,\bar{Z})=(z,\bar{z}))\nonumber \\
&\phantom{=}\cdot \Pr((Z,\bar{Z})=(z,\bar{z})|(Y,\bar{Y})=(y,\bar{y})).
\end{align}
Since $X,Y,Z$ are independent from $\bar{X},\bar{Y},\bar{Z}$, we have 
\begin{align}
\eqref{equ:proof of superposition_2} 
&=\Pr(X=x|Z=z)\Pr(\bar{X}=\bar{x}|\bar{Z}=\bar{z})\nonumber \\
&\phantom{=}\cdot \Pr(Z=z|Y=y) \Pr(\bar{Z}=\bar{z}|\bar{Y}=\bar{y})\nonumber \\
&\overset{\eqref{equation:X'Z}}{=}\Pr(X'=x,Z=z|Y=y) \Pr(\bar{X}',\bar{Z}=\bar{x},\bar{z}|\bar{Y}=\bar{y}), \nonumber
\end{align}
which implies that $X',Z|Y=y$ and $\bar{X}',\bar{Z}|\bar{Y}=\bar{y}$ are also independent for every~$y,\bar{y}\in\mathcal{Y}\times \bar{\mathcal{Y}}$. Therefore, we have 
\begin{align*}
&=  I(X',Z|Y) + I(\bar{X}',\bar{Z}|\bar{Y}) \nonumber \\
&= \operatorname{Un}(X'\to Z|Y) + \operatorname{Un}(\bar{X}'\to \bar{Z}|\bar{Y})\qedhere
\end{align*}

\end{proof}

\subsection{Proof of Lemma \ref{le:continuity of do operation}.}
\label{proof:continuity of do operation}

\begin{proof}
Recall that Lemma~\ref{le:redefinition of redundant information} states that $\operatorname{Red}(X,Y\to Z)=I(X';Y)$. Therefore, since $I(X';Y)$ is a continuous function of $\mathcal{D}_{X',Y}$, it suffices to prove that 
the mapping~$\mathcal{F}(\mathcal{D}_{X,Y,Z})=\mathcal{D}_{X',Y}$ is continuous, which holds by~\eqref{equation:X'Z}.
\end{proof}

\subsection{Proof of Lemma \ref{le:Independent Identity}.}
\label{proof:Independent Identity}
\begin{proof}
By Def.~\ref{definition:un}, we have~$\operatorname{Un}(X \to Z | Y) = I(X';Z|Y)$, where 
\begin{align}
\label{equ:proof:Independent Identity_1}
&\Pr(X'=x,Z=z|Y=y) \nonumber \\
&=\Pr(X=x|Z=z)\Pr(Z=z|Y=y).
\end{align}
Since~$Z=(X,Y)$, Eq.~\eqref{equ:proof:Independent Identity_1} is zero whenever~$z\ne (x,y)$, and otherwise
\begin{align}
\label{equ:proof:Independent Identity_2}
\eqref{equ:proof:Independent Identity_1} &=\Pr(X=x|(X,Y)=(x,y))\Pr((X,Y)=(x,y)|Y=y)\nonumber \\
&=\Pr(X=x|Y=y).
\end{align}
Also, since $X$ and $Y$ are independent, we have 
\begin{align*}
\eqref{equ:proof:Independent Identity_2} = \Pr(X=x).
\end{align*}
Therefore, we have 
\begin{align*}
\operatorname{Un}(X \to Z | Y) = I(X';Z|Y) = H(X) = I(X;Z),
\end{align*}
which by Definition \ref{definition:red} implies that
\begin{align*}
\operatorname{Red}(X ,Y \to Z) &= I(X;Z) -  \operatorname{Un}(X \to Z | Y)=0.\qedhere
\end{align*}
\end{proof}

}{}
\end{document}